\documentclass[runningheads,a4paper]{llncs}

\usepackage[latin1]{inputenc}
\usepackage{amssymb,amsmath}
\usepackage{algorithm,algorithmic}
\usepackage{geometry}
\usepackage{subfig}	
\usepackage{graphicx}
\usepackage{url}
\urldef{\mailsa}\path|mbaye@ept.sn,|
\urldef{\mailsb}\path|aaciss@ept.sn, |
\urldef{\mailsc}\path|oniang@ept.sn|
\urldef{\arxiv}\path |http://www.arxiv.org|
\newcommand{\keywords}[1]{\par\addvspace\baselineskip
\noindent\keywordname\enspace\ignorespaces#1}

\begin{document}

\title{A Lightweight Identification Protocol for Embedded Devices}
\titlerunning{A Lightweight Identification Protocol for Embedded Devices}

\author{Abdoulaye Mbaye, Abdoul Aziz Ciss and Oumar Niang}
\authorrunning{Abdoulaye Mbaye and Abdoul Aziz Ciss and Oumar Niang}
\institute{Laboratoire de Traitement de l'Information et Systèmes Intelligents,
\\ \'Ecole Polytechnique de Thiès \\
\mailsa \mailsb \mailsc}
\maketitle
\bibliographystyle{plain}

\begin{abstract}
The task of this paper is to introduce a new lightweight identification protocol based on biometric data and elliptic curves. In fact, we combine biometric data and asymetric cryptography, namely elliptic curves and standard tools to design a multifactor identification protocol. Our scheme is light,  very fast, secure and robust against all the known attacks on identification protocol. Therefore, one can use it in any constraint device such as embedded systems.
.
\end{abstract}

\section{Introduction}
In the era of ubiquitous computing, the need to know the identity and authenticity of the entity with which we communicate and ensure the privacy of users has become essential. This new era has brought a new cryptography for low-resource devices such as sensor networks, RFID tags, smartphone, smart cards \ldots . This  is called  'lightweight cryptography'. Design an effective protocol that meets all the security requirements of small devices with limited resources is a challenge.
According to \cite{mobahat2010authentication}, it  is possible but difficult to find a compromise between performance and security, performance and cost or cost and security. Thus, a large number of researchers in recent years that engage in the search for solutions Security and privacy in these smart devices \cite{saarinen2012all}.

Much of the research on the safety of these devices focus more on identification protocols \cite{batina2012hierarchical,vajda2003lightweight}.
For instance, when several of these heterogeneous objects communicate with each other or with other computer systems, they do not have a direct view of their interlocutors. Thus, to identify and recognize the different actors in a communication it is important to identify them.

Nevertheless, the user's identity becomes a key factor of a system, it is necessary to ensure the legitimacy of the claimed identity, in order to avoid situations such as identity theft or stealing identity and non-repudiation. We must not only identify the entity, but also authenticate, ensure safely of his identity.

Indeed, in \cite{vajda2003lightweight} the authors propose a set of burglar extremely lightweight authentication schemes such as challenge-response using logic and basic arithmetic operations, which can be used in RFID systems to authenticate tags. Some protocols of the latter were found to be weak and cryptanalyzed  \cite{saarinen2012all,defend2007cryptanalysis}.

In RFID systems, sensor networks, smart cards, \ ldots, most of the protocols used for identification using standard cryptographic primitives that are most often based on the random oracle model \cite{cho2011securing,kitsos2008rfid,franklin2008cryptology,guo2011rfid,chien2009study,DavidRodrigoJavier,MashatanStinson2008,MashatanStinson2007,poschmann:lightweight}.
With the limitations of these devices, the solid support of cryptographic primitives is not optimal. The use of standard cryptographic hash functions is beyond the capabilities of these devices. Therefore, there is a strong need for new lightweight cryptographic primitives that can be supported by these tools to limited resources such as deterministic randomness extractors on elliptic curves \cite{bouillaguet2011etudes,ciss2014two,ciss2011randomness}.

The Schnorr's identification protocol  is used in environments with limited resources since it offers a good level of security while respecting the constraints posed by these technologies \cite{schnorr1990efficient}. The security of this protocol is based on the difficulty of extracting the discrete logarithm in suitable finite fields.

In this paper, we propose a new lightweight identification protocol for limited-resources devices. In fact, our protocol is designed for  any elliptic curve defined over the finite field $K$, where $K$ can be a prime or a binary field. However, we used the Edwards elliptic curve to gain complexity and contract side channel attacks \cite{hisil2008twisted,bernstein2008twisted}. We also use deterministic encodings to map random elements of the base field to a point of the curve  \cite{tibouchi2014elligator} and also deterministic randomness extractors to derive suitably random bit-string from a a random point of the elliptic curve  \cite{bouillaguet2011etudes,ciss2014two,ciss2011randomness}. Our protocol is secure under the standard model assuming the discrete logarithm problem is hard.
The rest of the paper is organized as follows

\section{Preliminaries}
\subsection{Elliptic curves}
\subsubsection{Definition}

An elliptic curve $E$ over a field $K$ can be described as the subset of $K\times K$ satisfying the equation $$y^2+a_1xy+a_3y=x^3+a_2x^2+a_4x+a_6$$
for a given $a_1,a_2,a_3,a_4,a_5,a_6\in K$, along with another special point "at infinity'' $\mathcal{O}$. An additional demand is that the curve be "smooth'',
which means that the partial derivatives of the curve has no common zeros. This can be reduced to checking that some invariant value $\Delta$ (the \emph{discriminant} of the curve)
which is calculated from the coefficients is not zero.

\begin{figure}
   \centering
   \includegraphics[scale = 0.5]{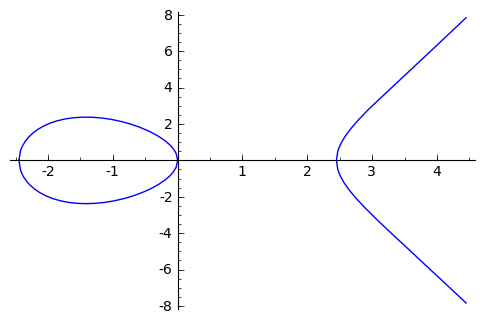}
   \caption{An elliptic curve (corresponding to the equation $y^2=x^3-6x$)}
\end{figure}

Depending on the characteristic of $K$, the above equation can be simplified. There are three cases to consider:
\begin{enumerate}
  \item When $\mathrm{char}(K)\ne 2,3$ the equation can be simplified to $y^2=x^3+ax+b$ with $a,b\in K$.
  \item When $\mathrm{char}(K) = 2$ and $a_1 \ne 0$, the equation can be simplified to $y^2+xy=x^3+ax^2+b$ with $a,b\in K$. This curve is said to be \emph{non-supersingular}. If $a_1 = 0$, the equation can be simplified to $y^2+cy=x^3+ax+b$ with $a,b,c\in K$. This curve is said to be \emph{supersingular}.
  \item When $\mathrm{char}(K) = 3$ and $a_1^2\ne -a_2$, the equation can be simplified to $y^2=x^3+ax^2+b$ with $a,b\in K$. This curve is said to be \emph{non-supersingular}. If $a_1^2=-a_2$, the equation can be simplified to $y^2=x^3+ax+b$ with $a,b\in K$. This curve is said to be \emph{supersingular}.
\end{enumerate}

\subsubsection{Group Structure.}
The structure of the group of an elliptic curve highly depends on the underlying field. Over a finite field $K=\mathbb{F}_q$ ($q$ a power of a prime), the order of the curve is, of course, finite. However, it is not easy to know in general what is the exact number of points on the curve.  A reasonable estimate is given by Hasse's theorem, which shows that the order is roughly $q+1$, with a square-root bound on the error. More precisely, if $\#E(\mathbb{F}_q)$ is the number of points on the curve $E$ over $\mathbb{F}_q$, then $|\#E(\mathbb{F}_q)-(q+1)|\le 2\sqrt{q}$

Knowing the order of the curve is very important when using curves for cryptographic purposes. There exist algorithms for computing the order, but
they are quite complicated and not very efficient. Thus, sometimes a different approach is used - first the order is chosen, and then a suitable curve
for that order is generated.
\\~~\\
\textbf{Addition Formulas - The Prime Field Case.}
We consider now the case $K=\mathbb{F}_p$ for $p>3$ prime (however, for any $K$ such that $char K\ne 2,3$ our description is essentially the same).
In this case let $E$ be a curve defined by $y^2=x^3+ax+b$. Given two points on the curve $P=(x_1,y_1),Q=(x_2,y_2)$ we first consider the case $x_1\ne x_2$
(and so $P\ne\pm Q$). The line passing both in $P$ and $Q$ has slope $\frac{y_2-y_1}{x_2-x_1}$ and so its equation is $y=y_1+\left(\frac{y_2-y_1}{x_2-x_1}\right)(x-x_1)$.

Substituting $y$ into the curve's equation yields a complicated equation of degree 3 in $x$; however, since we already know of two roots to the equation $x_1,x_2$ we can find the third one $x_3$ by using the fact that $-(x_1+x_2+x_3)$ equals the coefficient of $x^2$ in the equation, which is exactly the square of the slope, $\left(\frac{y_2-y_1}{x_2-x_1}\right)$ Hence we have $$x_3=\left(\frac{y_2-y_1}{x_2-x_1}\right)^2-x_1-x_2$$

Given $x_3$, finding $y_3$ is easy, and we see that $$y_3=\left(\frac{y_2-y_1}{x_2-x_1}\right)(x_1-x_3)-y_1$$

When adding a point $P=(x_1,y_1)$ to itself the only difference is the calculation of the slope, which in this case is the tangent to the curve at $x_1$, hence
the implicit derivative $\frac{3x_1^2+a}{2y_1}$. This yields the equations

$$x_3=\left(\frac{3x_1^2+a}{2y_1}\right)^2-x_1-x_2$$

and

$$y_3=\left(\frac{3x_1^2+a}{2y_1}\right)(x_1-x_3)-y_1$$

~~\\
\textbf{Addition Formulas - The Binary Field Case.}
In the binary case $K=\mathbb{F}_2$, the curve $E$ is defined by $y^2+xy=x^3+ax^2+b$. First note that the negative of a point $S=(x,y)$ is defined as $-S=(x,x+y)$, and it can be checked that the point $-P$ indeed lies on $E$.
For two points $P=(x_1,y_1),Q=(x_2,y_2)$ (denote $P+Q=(x_3,y_3)$), the easiest cases are $P=\infty$ or $Q=\infty$. By the group law, it holds that $P+\infty=P$ and $\infty+Q=Q$. In the case of $P=-Q$, the result is $P+Q=\infty$.

The remaining two cases are $P=Q$ and $P\ne Q$.

If $P=Q$, we need to double the coordinate, where $\lambda=x_1+\frac{y1}{x1}$ is the slope of the tangent, analogue to the prime case. As resulting x-coordinate one gets $x_3=a+\lambda^2+\lambda$.

If $P\ne Q$, the slope of the connecting line is $\frac{y_1+y_2}{x_1+x_2}$ (note that addition and subtraction are the same in binary field arithmetic). It holds that $x_3=a+\lambda^2+\lambda+x_2+x_1$ (in the case $P=Q$, $x_2+x_1=x_1+x_1=0$).

Finally in both cases the y-coordinate is computed as $y_3=(x_2+x_3)\lambda+x_3+y_2$. The overall algorithm derived from the formula requires 3 multiplications (one is squaring), one multiplicative field inversion, and 9 (or 6 in case of doubling) additions. While the additions are very cheap operations, computing the inverse could be avoided by using projective coordinates as described in \ref{Projective Coordinates}.

\subsubsection{Point Multiplication.}
Scalar point multiplication is a frequently used operation in the cryptosystem we implemented and it is of big importance to implement this efficiently. We want to compute $kP=Q$ for a scalar $k$ and a point $P$. Using only addition, one could compute this as $\underbrace{P+...+P}_{k}$, with a complexity of $O(k)$. A better way is to use the additive equivalent of the \emph{square-and-multiply} algorithm. There, $k$ is processed bit-wise. In each iteration, the result is doubled, and if the bit equals one, $P$ is added. Let $[k_{\log{k}},...,k_1]$ be the bit-representation of $k$, with $k_1$ as the least significant bit:

\begin{enumerate}
 \item $Q=\infty$
 \item for $i=1..\log{k}$:
 \item $Q=2Q$
 \item if $k_i=1$: $Q=Q+P$
\end{enumerate}

The complexity of this algorithm is $O(\log{k})$. There are some optimizations to this algorithm, most notably the replacement of the binary representation of $k$ with its \emph{non-adjacent form} which results in a shorter representation.\\~~\\
Note that for sake of efficiency and the need to avoid  the most known attacks on elliptic curve point multiplication, new family of elliptic curves were introduced, namely Edwards elliptic curves [cite Edwards], Hessian Elliptic curves [cite Hessian], Huff model [cite Huff].

\subsubsection{Randomness extractors.}~~\\
For $q = p$ a prime number $>5$ let's recall the extractor of Chevalier \emph{et al.} in \cite{op}
\begin{definition}
Let $E$  be an elliptic curve defined over a finite field $\mathbb{F}_p$, for a prime $p>2$. Let $G$ be a subgroup of $E(\mathbb{F}_p)$ and let $k$ be a positive integer. Define the function
\begin{align*}
\mathcal{L}_k : G& \longrightarrow \{0, 1\}^k\\
				P& \longmapsto \mathrm{lsb}_k(\mathrm{x}(P))
\end{align*}
\end{definition}
The following lemmas state that $\mathcal{L}_k$ is a deterministic randomness extractor for the elliptic curve $E$

\begin{lemma}
Let $p$ be a $n$-bit prime, $G$ a subgroup of $E(\mathbb{F}_p)$ of cardinal $q$ generated by a point $P_0$, $q$ being an $l$-bit prime, $U_G$ a random variable uniformly distributed in $G$ and $k$ a positive integer. Then
$$\Delta(\mathcal{L}_k(U_G), U_k)\leq 2^{(k+n+\log_2(n))/2 + 3 - l},$$
where $U_k$ is the uniform distribution in $\{0, 1\}^k$.
\end{lemma}
~~\\
\emph{Proof. } See \cite{op}.
\begin{corollary}
Let $e$ be a positive integer and suppose that
$$k \leq 2l -(n+2e+\log_2(n)+6).$$ Then $\mathcal{L}_k$ is a $(U_G, 2^{-e})$-deterministic extractor
\end{corollary}

Consider now the finite field $\mathbb{F}_{p^n}$, where $p > 5$ is prime and $n$ is a positive integer. Then $\mathbb{F}_{p^n}$ is a $n$-dimensional vector space over $\mathbb{F}_p$. Let $\{\alpha_1, \alpha_2, \ldots, \alpha_n\}$ be a basis of $\mathbb{F}_{p^n}$ over $\mathbb{F}_p$. That means, every element $x$ of $\mathbb{F}_{p^n}$ can be represented in the form $x=x_1\alpha_1 + x_2\alpha_2 +\ldots +x_n\alpha_n$, where $x_i\in\mathbb{F}_{p^n}$.
Let $E$ be the elliptic curve over $\mathbb{F}_{p^n}$ defined by the Weierstrass equation
$$y^2+(a_1 x+a_3)y=x^3+a_2 x^2+a_4 x+ a_6.$$

The extractor $\mathcal{D}_k$, where $k$ is a positive integer less than $n$,  for a given point $P$ on $E(\mathbb{F}_{p^n})$, outputs the $k$ first $\mathbb{F}_p$-coordinates of the abscissa of the point $P$.
\begin{definition}
Let $G$ be a subgroup of $E(\mathbb{F}_{p^n})$ and $k$ a positive integer less than $n$. Define the function $\mathcal{D}_k$ i
\begin{align*}
\mathcal{D}_k : G & \longrightarrow \mathbb{F}_{p^k}\\P=(x,y) &
\longmapsto (x_1, x_2, \ldots, x_k)
\end{align*}
where $x\in \mathbb{F}_{p^n}$ is represented as $x=x_1\alpha_1 + x_2\alpha_2 +\ldots +x_n\alpha_n$, and $x_i\in\mathbb{F}_{p^n}$.
\end{definition}

\begin{lemma}
Let $E$ be an elliptic curve defined over $\mathbb{F}_{q}$, whit $q = p^n$  and let $G$ be a subgroup of $E(\mathbb{F}_{p^n})$. Let $\mathcal{D}_k$ be the function defined above. Then,
$$\mathrm{Col}(\mathcal{D}_k(U_G)\leq \frac{1}{p^k} + \frac{4\sqrt{q}}{|G|^2}$$
and
$$\Delta(\mathcal{D}_k(U_G), U_{\mathbb{F}_{p^k}})\leq \displaystyle \frac{2\sqrt{p^{n+k}}}{|G|}$$
where $U_G$ is uniformly distributed in $G$ and $U_{\mathbb{F}_{p^k}}$ is the uniform distribution in
$\mathbb{F}_{p^k}$.
\end{lemma}
~~\\
\emph{Proof. } See \cite{ciss2}

\begin{lemma}
Let $p>2$ be a prime and  $E(\mathbb{F}_{p^n})$ be an elliptic curve over $\mathbb{F}_{p^n}$ and  $G \subset E(\mathbb{F}_{p^n})$ be a multiplicative subgroup of order $r$ with $|r|=t$ and $|p|=m$ and let $U_G$ be the uniform distribution in $G$. If  $e>1$ is an integer  and $k>1$ is an integer such that
\begin{equation*}
k\leq \frac{2t - 2e - nm - 4}{m},
\end{equation*}
then $\mathcal{D}_k$ is a $(\mathbb{F}_{p}^k, 2^{-e})$-deterministic randomness extractor over the elliptic curve $E(\mathbb{F}_{p^n})$.
\end{lemma}
~~\\
\emph{Proof. } See \cite{ciss2}

\subsubsection{Hashing into elliptic curves.}
Many cryptographic schemes based on elliptic curves require efficient hashing of finite field elements into a given elliptic curve
Password based authentication protocols give a context where hashing into elliptic curves is sometimes required. For instance, the SPEKE (Simple Password Key Exchange) [6] and the PAK (Password Authenticated Key Exchange) [4] protocols both require a hash algorithm to map the password into a point of the curve. The Boneh-Franklin identity-based encryption scheme is another context of use of hashing into elliptic curves.
See [3] for a short survey of applications.

Many hash functions into elliptic curves were proposed. In 2009, Icart introduced and studied such hash functions [3]. More precisely, let $E_{a, b}$ be an  elliptic curve defined over the finite field $\mathbb{F}_q$, with $p = \mathrm{char}(\mathbb{F}_q) \geq 5$, given by the Weierstrass equation
$$y^2 = x^3 + ax + b,$$
where $a, b \in \mathbb{F}_q$. Let $E(\mathbb{F}_q)$ denote the set of $\mathbb{F}_q$-rational points  on $E_{a,b}$, including the point at infinity $\mathcal{O}$. \\
For $q \equiv 2 \ \mathrm{mod} \ 3$, Icart [3] proposed the following map
\begin{align*}
f_{a, b} : \mathbb{F}_q^* &\longrightarrow E_{a, b}(\mathbb{F}_q)\\
                        u & \longmapsto (x, y)
\end{align*}
where $$x = \left(v^2 - b - \frac{u^2}{27}\right)^{1/3} + \frac{u^2}{3} \ \ \ \ \ \text{and} \ \ \ \ \ y = ux + v,$$
with $$v = \frac{3a - u^4}{6u},$$
and $f(u) = \mathcal{O}$ if $u=0$.
\\~~\\
Many other encoding maps into (hyper)elliptic curves were proposed, namely for Edwards curves [elligator, square elligator], Hessian Curves  [Cite hash Hess], hyperelliptic curves [indifferent...]
\\~~\\
\textbf{Hashing into Edwards curves : Elligator \small{(Bernstein \emph{et. al} 2013)}}
Let  $E_d$ be an Edwards curve defined over $\mathbb{F}_q$  by $$x^2 + y^2 = 1 + dx^2y^2,$$ with $d \neq 0, 1$ is not a square. Define the quadratic character
$$\chi : \mathbb{F}_q \longrightarrow \mathbb{F}_q, \ \ a\longmapsto a^{(q-1)/2}$$
If $a$ is a non-zero square, then $\chi(a) = 1$; if $a$ is a non-square, then $\chi(a) = -1$; if $a = 0$, then $\chi(a) = 0$
\begin{lemma}
\label{hashed}
Let $q$ be a prime power congruent to $3$ modulo $4$. Let $s$ be a nonzero element of $\mathbb{F}_q$ with $(s^2 - 2)(s^2+2) \neq 0$. Define $c = 2/ s^2$. Then $c(c - 1)(c + 1) \neq 0$. Define $r = c + 1/c$ and $d  = -(c+1)^2/(c-1)^2$. Then, $r\neq 0$, and $d$ is not a square.
The following elements of $\mathbb{F}_q$ are defined for each $t\in \mathbb{F}_q - \{\pm 1\}$ :
$$u = (1-t)/(1 + t),    \ \ v = u^5 + (r^2 - 2)u^3 + u$$
$$X = \chi(v)u, \ \ Y = (\chi(v)v)^{(q + 1)/4} \chi(v)\chi(u^2+1/c^2)$$
$$x = (c-1)aX(1 + X)/Y, \ \ y = (rX - (1 + X^2))/(rX + (1+X^2))$$
Furthermore $x^2 + y^2 = 1 + dx^2y^2$
\end{lemma}

\begin{definition}\label{defhash}
In the situation of Lemma \ref{hashed}, the decoding function for the complete Edwards curve $E : x^2 + y^2 = 1 + dx^2y^2$ is the function $\phi : \mathbb{F}_q \longrightarrow E(\mathbb{F}_q)$ defined as follows : $\phi(\pm1) = (0,1)$; if $t \notin \{\pm 1\}$ then $\phi(t) = (x, y)$.
\end{definition}

\begin{lemma}
In the situation of Definition \ref{defhash}, assume that $q$ is prime, and define $b = \lfloor \log_2(q)\rfloor$. Let $\sigma : \{0,1\}^b\longrightarrow \mathbb{F}_q$ by $\sigma(\tau_0, \tau_1, \ldots, \tau_{b-1}) = \sum_{i}\tau_i 2^i$. Define $S = \sigma^{-1}(\{0,1,2, \ldots, (q-1)/2\})$. Let $\iota$ be the function $\iota : S\longrightarrow E(\mathbb{F}_q)$, $\iota(\tau) = \phi(\sigma(\tau))$. Then $\#S = (q+1)/2$; $\iota$ is an injective map from $S$ to $E(\mathbb{F}_q)$ and $\iota(S) = \phi(\mathbb{F}_q)$.
\end{lemma}

\subsection{Identification protocols}
It is well known that patterns of identification are very important to cryptography and computer security.\\
An authentication protocol is a part of an identification pattern. An identification system is composed of two protocols appointed registration and identification between two parties. In an identification scheme basic, recording ends with both parties sharing a secret key that both need to store safely. In general, authentication protocols are based on one or more of the following factors.
\begin{enumerate}
\item What you are ? Biometrics.
\item What you have ? Smart cards, SIM cards, or similar computers working.
\item What you know ? Passwords, PIN codes, secret keys.
\end{enumerate}
In this section, we focus on purely cryptographic authentication protocols, in which a successful prover must know where to have some secret key. The general objective of authentication protocols in cryptographic constructions is to reduce the computation time for the prover and/or verifier.
These protocols are usually three to four rounds.

The aim of the adversary in an identification scheme is identity theft or to attack the system so that it behaves like a cheater prover and attempts to identify the honest verifier. With the existence of the opponent attempting to impersonate three common attacks are generally considered, passive attacks, active and simultaneous attacks. Therefore, the security against these attacks has become a major concern in cryptography, where the analysis and creation of identification systems are widely studied.

The principle of zero-knowledge is used as the identification scheme. It is used especially in the context of digital signatures, often on smart card. Many protocols exist, and many variations evolving despite the findings of some security vulnerabilities.

Fiat and Shamir \cite {fiat1987prove} have faith first proposed an interactive identification scheme for zero-knowledge, also involving two parts: a prover $P$ and a  verifier $V$, where the prover convince spineless the verifier knowledge of a secret without revealing any information other than the fact that he knows the secret.

Then a lot of identification based on zero-knowledge protocol knowledge have been proposed: we include in particular Feige-Fiat-Shamir (FFS) \cite{feige1988zero} , Schnorr \cite{schnorr1990efficient} and Guillou-Quisquater(GQ) or Ohta-Okamoto \cite{ohta1990modification}. We detail the Schnorr identification on elliptic curves in the following protocol.

The Guillou-Quisquater scheme (QM) is based on the problem of inversion of RSA. Very close to the protocol Schnorr, they show that their scheme is secure against passive attacks, provided the difficulty of factorization.

Ohta and Okamoto (OO) present a modification of the Fiat-Shamir scheme based on the difficulty of extracting the $ i^{th} $ root and they prove their scheme is secure as the Fiat-Shamir scheme.

In \cite{okamoto1993provably}, Okamoto's identification scheme exists in  three version. The first is based on the discrete logarithm problem, the second one on the RSA problem and the last on the factorization problem. All these schemes have been proven secure against the passive attack.

One of the simplest and most frequently used knowledge proofs, proof of knowledge of a discrete logarithm is due to Schnorr. It provides evidence of a zero-knowledge prover is in possession of also known as proof of knowledge given secret. A challenge $ c $ is used to verify the identity of the prover without ever revealing the secret $s$. \\
It may well be used in the environment forced on all devices with limited resources, since it offers a good level of security while respecting the constraints posed by these technologies. \\
In Figure \ref{schnor}, we present the native version of Schnorr identification scheme over an elliptic curve

\begin{figure}[hbtp!]
\label{schnor}
\caption{ECSchnorr identification scheme.}
\centering
\includegraphics[scale=1]{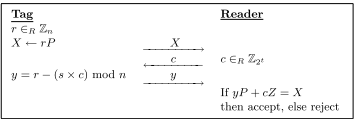}
\end{figure}
Where $n$ denotes the order of the elliptic curve $ s $ as the secret key, $r$ as Escrow, $ c $ as a challenge and $y$ as challenge-response. All points of the elliptic curve are given in capitals where $ P $ is the base point of the curve $X$ is the point of control, and $Z$ is the point that corresponds to the public key. \\ Description Protocol
\begin{enumerate}
\item {Key Generation:} To generate a key pair, the prover sends a
hazard $ s \int_{R}  Z_n$. Its private key is $s$ and the public key is $ Z  = sP$.
\item {Commitment phase: To engage}, the prover sends a random number $r \ int_{R}  Z_n$ and calculates $X = rP$. It then sends $X$ to the verifier.
\item {Challenge phase:} The verifier chooses a random number $ c \int_{R} Z_{2^t} $ and sends the prover.
\item {Response phase:} The prover then computes $ y = r - (s*c) \ \textrm{mod}\ n $, it returns to the verifier.
\item {Verification:} The verifier checks if $X  = yP + cZ$ which case the identification is successful.
\end{enumerate}

The protocol is secure against passive opponents, but it is not protected against active and concurrent attacks.

Many patterns of light identifications were proposed in the literature for use on lightweight cryptographic devices as in \cite{BiblioRFIDGA}.

\section{Our protocol}
This section introduces the main result of the paper. In fact, we propose a new lightweight weak-strong identification protocol secure under the standard model, using elliptic curves and biometric data. The idea is to combine two kind of identification schemes : biometric based authentication and challenge-response identification.

Let $E$ be an elliptic curve defined over the finite field $\mathbb{F}_q$, where $q$ is a prime power. Denote $E(\mathbb{F}_q)$ be the set of $\mathbb{F}_q$-rational points over $E$. Let
\begin{align*}
h : \mathcal{B} &\longrightarrow E(\mathbb{F}_q)\\
                u &\longmapsto Q = (x, y)
\end{align*}
be a hash function into the elliptic curve, where $\mathcal{B}$ denote the set of biometric data which can be viewed as bit-strings. $h$ can be the Icart function for elliptic curves in Weierstrass form or the Elligator when working with Edwards elliptic curves.\\~~\\
Alice, the claimant (the prover) maps her biometric data $b$ (a bit-string consisting the combination of fingerprint and retinal scan or voice scan) to a point $B$ of the elliptic curve $E$ using function $h$. She also select a random integer $\alpha$, computes $s = Ext_k(\alpha B)$, where $Ext_k$ may be the function $\mathcal{L}_k$ or $\mathcal{D}_k$, and sends the couple $(B, s)$ to the verifier (Bob) through asymetric cryptography methods. For instance, $B$ can be encrypt using the public key of the verifier. Therefore, the claimant and the verifier share the same secret $(B,s)$, which is used for weak authentication. Here, $B$ play the role of a password.\\~~\\
\textbf{Setup phase. }The claimant
\begin{enumerate}
\item Chooses a random point $P\in E(\mathbb{F}_q)$ with order $l$ and computes  .
\item Computes $$C = \alpha P + B$$
\item Sends $P$ and $C$ to the verifier
\item Keeps secrete the values $\alpha$
\end{enumerate}
~~\\
\textbf{Identification phase.} Let $\mathcal{C}$ be the claimant and $\mathcal{V}$ be the verifier
\begin{enumerate}
\item $\mathcal{C}$ chooses a random value $r <_{R} l$, computes $D = rP - \alpha B$ and sends $D$ to $\mathcal{V}$
\item $\mathcal{V}$ chooses and sends to $\mathcal{C}$ a random value $e \in\{1, \ldots, 2^{t-1}\}$, where $t$ is the security parameter
\item $\mathcal{C}$ computes $y = r - e\alpha$ and sends $y$ to $\mathcal{V}$
\item Upon receiving $y$, $\mathcal{V}$ verifies if $$Ext_k(yP - D - e(B-C)) = s$$
\end{enumerate}

\begin{theorem}
The verification process is correct.
\end{theorem}

\begin{proof}
In fact,
\begin{align*}
yP - D - e(B-C) &= (r-e\alpha)P - (rP - \alpha B) -e(-\alpha P)\\
                &= \alpha B
\end{align*}
Thus, $$Ext_k(yP - D - e(B-C) = Ext_k(\alpha B) = s$$
\end{proof}

\begin{theorem}
Assume the hardness of DLP and the randomness of the biometric data, then the proposed  identification protocol is secure against direct impersonation attacks.

Specifically,
\end{theorem}

\begin{proof}
The goal of the adversary is to impersonate the prover. To achieve this task, he has to solve the equation $y = r - eb$, where $b$ is the unknown. The value of $b$ can be derived from the public parameter $C = \alpha P + B$, ie. $\alpha P = C-B$ if the adversary knows the value of $B$ and can solve also the DLP over $E$ at the same time. In other words, as far as the DLP is hard, our proposition is secure.
\end{proof}

\section*{Conclusion}
Based on the difficulty of the discrete logarithm and the rarity of an authentication protocol based on the standard model for lightweight cryptography, we proposed a zero-knowledge protocol identification by combining for example biometrics which provides better efficiency. As proposed in Section 3, this protocol uses a hash function or Elligator, Icart and randomness extractor. This allows to increase the performance, the safety and the cost in the identification. We also showed that our protocol is secure against impersonation of active attacks. Therefore, the proposed identification system is more desirable than the existing systems for these light technology.
\bibliography{biblio}

\end{document}